\documentclass{amsart}
\pdfoutput=1
\usepackage[colorlinks]{hyperref}
\usepackage{multirow}
\usepackage{booktabs}
\usepackage{mathtools}
\usepackage{mathrsfs}
\usepackage[all]{xy}

\newtheorem{theorem}{Theorem}[section]
\newtheorem{proposition}[theorem]{Proposition}

\newtheorem{question}[theorem]{Question}
\newtheorem{definition}[theorem]{Definition}

\theoremstyle{remark}
\newtheorem{remark}[theorem]{Remark}
\newtheorem{example}[theorem]{Example}

\DeclareMathOperator{\Aut}{Aut}

\newcommand{\FF}{{\mathbf{F}}}
\newcommand{\PP}{{\mathbf{P}}}
\newcommand{\ZZ}{{\mathbf{Z}}}
\renewcommand{\aa}{{\mathbf{a}}} 
\newcommand{\bb}{{\mathbf{b}}}
\newcommand{\cc}{{\mathbf{c}}}

\newcommand{\calC}{{\mathcal{C}}}
\newcommand{\calL}{{\mathcal{L}}}

\newcommand{\Fq}{\FF_q}
\newcommand{\scrO}{{\mathscr{O}}}

\newcommand{\BibTeX}{{\textrm{B}\kern-.05em{\textsc{i}\kern-.025em \textsc{b}}\kern-.08em \TeX}}

\newcommand{\Vladut}{Vl\u adu\c t}

\newcommand\lowtilde{\lower0.7ex\hbox{\textasciitilde}}

\begin{document}

\title[Locally recoverable codes]{Locally recoverable codes 
\\from algebraic curves and surfaces}

\author[Barg]{Alexander Barg}               
\address{Department of Electrical and Computer Engineering, 
         Institute for Systems Research,
         University of Maryland, 
         College Park, MD 20742, USA and 
         Institute for Information Transmission Problems, 
         Russian Academy of Sciences, 
         Moscow, Russia}
\email{abarg@umd.edu}

\author[Haymaker]{\hbox{Kathryn Haymaker}}  
\address{Department of Mathematics and Statistics,
         Villanova University,
         800 Lancaster Avenue 
         Villanova, PA 19085, USA}
\email{kathryn.haymaker@villanova.edu}

\author[Howe]{\hbox{Everett W. Howe}} 
\address{Center for Communications Research,
         4320 Westerra Court,
         San Diego, CA 92121, USA}
\email{however@alumni.caltech.edu}
\urladdr{http://www.alumni.caltech.edu/\lowtilde{}however/}

\author[Matthews]{\hbox{Gretchen L. Matthews}}
\address{Department of Mathematical Sciences,
         Clemson University,
         Clemson, SC 29634, USA}
\email{gmatthe@clemson.edu}

\author[V\'arilly-Alvarado]{\hbox{Anthony V\'arilly-Alvarado}}
\address{Department of Mathematics, 
         Rice University,
         6100 S. Main St.   
         Houston, TX 77005, USA }
\email{av15@rice.edu}

\date{\today}
\keywords{}

\subjclass[2010]{Primary 94B27; Secondary 11G20, 11T71, 14G50, 94B05} 

\begin{abstract}
A \emph{locally recoverable code} is a code over a finite alphabet such that
the value of any single coordinate of a codeword can be recovered from the 
values of a small subset of other coordinates.  Building on work of Barg, Tamo,
and \Vladut, we present several constructions of locally recoverable codes from
algebraic curves and surfaces.
\end{abstract}

\maketitle

\tableofcontents
\section{Introduction}
\label{S:intro}
A code $\calC$ of length $n$ over a finite field~$\Fq$ is a subset of the linear space $\Fq^n$. 
The code is called linear if it forms a linear subspace of $\Fq^n$. The minimum distance of the code $d$ is the minimum
pairwise separation between two distinct elements of $\calC$ in the Hamming metric, and if the code $\calC$ is linear, then $d$ is equal to the minimum Hamming weight of a nonzero codeword of $\calC$.  We use the notation $(n,k,d)$ to refer to the parameters
of a linear $k$-dimensional code of length $n$ and minimum distance~$d$. 

Applications of codes in large-scale distributed storage systems motivate studying codes with locality constraints. One of the 
first code families of this kind was proposed in \cite{huang2007pyramid}. 
The concept of locality in codes was formalized in the following definition of locally recoverable codes, or LRCs.
\begin{definition}[LRCs, \cite{GopalanHuangEtAl2012}]\label{def:LRC} 
A code $\calC\subset \Fq^n$ is \emph{locally recoverable with locality $r$} if for every $i\in \{1,2,\dots,n\}$
there exists a subset $I_i\subset \{1,2,\dots,n\}\backslash \{i\}$
of cardinality at most $r$ and a function $\phi_i$ such that for every codeword $x\in\calC$ we have
   \begin{equation}\label{eq:def1}
   x_i=\phi_i(\{x_j,j\in I_i\}).
   \end{equation}
This definition can be also rephrased as follows. Given $a\in \Fq,$ consider the sets of codewords
   \begin{equation*}
   \calC(i,a)=\{x\in \calC: x_i=a\},\quad i\in\{1,2,\dots,n\}.
   \end{equation*}
    The code $\calC$ is said to have  \emph{locality} $r$ if for every $i$ there exists a subset $I_i\subset \{1,2,\dots,n\}\backslash i$ of cardinality at most $r$ 
    such that the restrictions of the sets $\calC(i,a)$ to
the coordinates in $I_i$ for different $a$ are disjoint:
 \begin{equation}\label{eq:def}
 \calC_{I_i}(i,a)\cap \calC_{I_i}(i,a')=\emptyset,\quad  a\ne a'.
 \end{equation}
\end{definition}

This definition applies to all codes without the linearity assumption, but throughout this paper we will assume that the codes are 
linear. The subset $I_i$ is called the \emph{recovery set} of the coordinate $i$. We note that every code with distance $\ge 2$ trivially has locality $r=k.$ At the same time, for applications of codes to storage systems we would like to have codes with small locality, high rate $k/n$, and
large distance $d$. The following Singleton-type inequality relating these quantities was proved in~\cite{GopalanHuangEtAl2012}:
\begin{equation}
\label{eq:singleton}
d\le n - k - \left\lceil\frac{k}{r}\right\rceil +2.
\end{equation}
If we use $r=k$ in this inequality, then this bound reduces to the classical Singleton bound of coding theory (e.g., \cite[p.\ 33]
{MaSl1977}). In the classical case a few code families are known to meet this bound, and the most well known among them is the family of Reed--Solomon codes.

Codes with locality whose parameters meet the bound \eqref{eq:singleton} with equality are called \emph{optimal LRCs}. Already in~\cite{GopalanHuangEtAl2012} 
it was shown that optimal LRCs do exist, and subsequent papers \cite{SilbersteinRawatEtAl2015,TPDMatroids} introduced several different constructions of optimal LRCs. These constructions relied on fields of size $q$ much larger than the code
length $n,$ for instance, $q=\exp(\Omega(n))$, while in the classical case, Reed--Solomon codes only require $q=n.$ This 
obstacle was removed in~\cite{TamoBarg2014}, which constructed optimal LRCs whose structure is analogous to the Reed--Solomon codes. In particular, the field size $q$ required for the codes of \cite{TamoBarg2014} is only slightly greater than the code length $n$.

The codes in~\cite{TamoBarg2014} are obtained from maps of degree $r+1$ from
$\PP^1_{\Fq}$ to $\PP^1_{\Fq}$. In~\cite{BargTamoEtAl2015}, the authors generalized 
this idea and constructed locally recoverable codes from morphisms of algebraic
curves; in particular, they produced examples from Hermitian curves and
Garcia--Stichtenoth curves.  In this paper we expand on the constructions 
of~\cite{BargTamoEtAl2015} to produce families of LRCs coming from a larger
variety of curves, as well as from higher-dimensional varieties.

There are several variations of the basic definition of LRC codes. One of them, also motivated by
applications, suggests to look for codes in which every coordinate $i$ has several disjoint
recovery sets, increasing the availability of data in storage. As pointed out in \cite{BargTamoEtAl2015}, 
codes with $t\ge 2$ recovery sets for every coordinate can be constructed from fiber products of curves. 
In this paper we construct codes with two recovery sets from fiber products of elliptic curves.
Recently this idea was further developed in \cite{Malmskog2016}, where the authors constructed examples of LRCs with multiple disjoint recovery sets from the Giulietti--Korchm{\'a}ros, Suzuki, and Hermitian curves. 
Codes with multiple recovery sets were also considered in \cite{CyclicLRC16}. The approach in that paper is different from the present work and relies on considering subfield subcodes of cyclic codes from \cite{TamoBarg2014}.

In conclusion let us remark that the problem of bounds for the parameters of LRCs, including asymptotic bounds, was studied in \cite{TamoBargEtAl2016,CaMa2015,Agarwal16}. In partiuclar,
these papers derived an asymptotic \emph{Gilbert--Varshamov type} bound for LRCs under the assumption of
constant $r$ and $n\to\infty.$ As shown in \cite {BargTamoEtAl2015}, 
the Garcia--Stichtenoth curves give families of LRCs whose parameters asymptotically exceed the Gilbert--Varshamov bound. In this paper we do not consider the asymptotic problem, focusing on finite code length.

In Section~\ref{sec:general} we present the general construction of LRCs from morphisms of varieties. In Section~\ref{sec:ECs} we construct a number of examples of codes from elliptic curves. In Section~\ref{sec:quartics} we construct codes with locality $r=3$ from plane quartics, and in Section~\ref{sec:highergenus} we study codes with locality $r=2$ from curves with automorphisms of order 3. In Section~\ref{sec:availability} we construct codes with 2 recovery sets from fiber products of elliptic curves. In Section~\ref{sec:higherdimension} we present a general construction of LRCs from algebraic surfaces and construct code families from
cubic surfaces, quartic K3 surfaces, and quintic surfaces. Finally, in 
Section~\ref{sec:conclusion} we collect and discuss the parameters of codes constructed in this paper.

\section{The general construction}
\label{sec:general}
In this section we present a variant of the construction
of~\cite[\S III]{BargTamoEtAl2015}, using slightly different assumptions
and notation.

Let $\varphi\colon X\to Y$ be a degree-$(r+1)$ morphism of 
projective smooth absolutely irreducible curves over~$K = \Fq$.
Let $Q_1, Q_2, \ldots, Q_s$ be points of $Y(K)$ that split completely 
in the cover $X\to Y$; that is, for each $Q_i$, there are $r+1$ points 
$P_{i,0}, P_{i,1}, \ldots, P_{i,r}$ in $X(K)$ that map to $Q_i$.

The map $\varphi$ induces an injection of function fields $\varphi^*\colon K(Y)\hookrightarrow K(X)$
that makes $K(X)$ a degree-$(r+1)$ extension of $K(Y)$.
Let $e_1, e_2, \ldots, e_r$ be elements of $K(X)$ that are linearly independent 
over $K(Y)$ and whose polar sets are disjoint from the $P_{i,j}$, and 
let $f_1, f_2,\ldots, f_t$ be elements of $K(Y)$ that are
linearly independent over~$K$ and whose polar sets are disjoint from the $Q_i$.

Let $S$ be the set of all pairs of integers $(i,j)$ with 
$1\le i\le r$ and $1\le j\le t$, and let $T$ denote the set of all
pairs of integers $(i,j)$ with $1\le i\le s$ and $0\le j\le r$.
We define a map $\gamma\colon K^S\to K^T$ as follows:
Given a vector $\aa = (a_{i,j})\in K^S$, let $f_\aa$
be the function
\[f_\aa := \sum_{i=1}^r e_i \sum_{j=1}^t a_{i,j} \varphi^* f_j,\]
and set
$\gamma(\aa) = \bb = (b_{i,j})\in K^T$ where
$b_{i,j} = f_\aa(P_{i,j})$ for all $i,j$ with $1\le i\le s$ and $0\le j\le r$.

Let $D$ be the smallest effective divisor on $X$ so that each product $e_i \varphi^* f_j$
lies in the Riemann--Roch space $\calL(D)$, and let $\delta$ be the degree of~$D$.
If $\delta < s(r+1)$ then the image $\calC$ of $\gamma$ is a linear code of dimension
$k = rt$, length $n = s(r+1)$, and minimum distance $d$ at least $s(r+1)-\delta$.

Under certain conditions, the code $\calC$ also has locality $r$, as we now explain.
For each $i =  1,\ldots, s$ let
\[H_i = \varphi^{-1}(Q_i) = \{ P_{i,j} \mid j \in\{0,\ldots,r\} \}.\]
We call the sets $H_i$ the \emph{helper sets} of the code.  

\begin{proposition}
\label{prop:matrices}
Let $i$ be an integer between $1$ and $s$, and suppose
every $r\times r$ submatrix of the matrix
\[
M := 
\begin{bmatrix}
e_1(P_{i,0}) & e_2(P_{i,0}) & \cdots & e_r(P_{i,0}) \\
e_1(P_{i,1}) & e_2(P_{i,1}) & \cdots & e_r(P_{i,1}) \\
\vdots  & \vdots & \ddots & \vdots \\
e_1(P_{i,r}) & e_2(P_{i,r}) & \cdots & e_r(P_{i,r}) \\
\end{bmatrix}
\]
is invertible.  Let $f$ be a function on $X$ such that $f = f_\aa$ for some
$\aa \in K^S$.  Then the value of $f$ on any point in the helper set $H_i$
can be calculated from the values of $f$ on the other points in the helper set.
\end{proposition}

\begin{proof}
Suppose the hypothesis of the proposition holds, and 
suppose $f = f_\aa$ for some element $\aa = (a_{u,v})$ of $K^S$.  We have
\[
f = \sum_{u=1}^r e_u \sum_{v=1}^t a_{u,v} \varphi^* f_v
\]
so for every $j=0,\ldots,r$ we have
\begin{align*}
f(P_{i,j}) &= \sum_{u=1}^r e_u(P_{i,\ell}) \sum_{v=1}^t a_{u,v} \varphi^* f_v(P_{i,\ell})\\
           &= \sum_{u=1}^r e_u(P_{i,\ell}) \sum_{v=1}^t a_{u,v} f_v(Q_i)\\
           &= \sum_{u=1}^r c_u e_u(P_{i,\ell})
\end{align*}
where 
\[ c_u = \sum_{v=1}^t a_{u,v} f_v(Q_i).\]
Since we know the values of $e_u(P_{i,j})$ for every $u$ and $j$, and since
every $r\times r$ submatrix of $M$ is invertible, we can calculate the values of 
the $c_u$ from the values of any subset of $r$ of the $f(P_{i,j})$.  Thus,
the value of any $f(P_{i,j})$ can be computed from the values of $f$ on
the $P_{i,\ell}$ with $\ell \ne j$.
\end{proof}

We see that by choosing a cover $\varphi\colon X\to Y$ together with
splitting points $Q_i \in Y$ and functions $f_i\in K(Y)$ and $e_i\in K(X)$
satisfying certain conditions, we wind up with a locally recoverable Goppa code.

In practice, to construct a code in this way from a given morphism
$\varphi\colon X \to Y$ one would most likely choose the divisor $D$ to begin
with, and then find functions $e_i$ on $X$ and $f_j$ on $Y$ such that each
product $e_i f_j$ lies in $\calL(D)$.

So where do we find a supply of covers $\varphi\colon X\to Y$?  One source
is to consider curves $X$ with nontrivial automorphism groups.  Suppose 
the automorphism group of $X$ contains a subgroup $G$ of order~$r+1$.
Then we can define $Y$ to be the quotient of $X$ by~$G$, and we have
a canonical Galois covering $\varphi\colon X\to Y$ with group~$G$.
This is the strategy we will use for most of the rest of this paper
when constructing LRCs from coverings of curves.
However, as we will see in section~\ref{sec:highergenus}, it is sometimes
better to fix a base curve $Y$ and consider families 
of Galois covers of $Y$ with group~$G$.

\section{Locally recoverable codes from elliptic curves}
\label{sec:ECs}

In this section we show how elliptic curves can be used to create codes
of locality $r$ for arbitrary~$r$.

Given an integer $r>1$, choose a prime power $q$ such that there is an 
elliptic curve $E$ over the finite field $K = \Fq$ such that $E(K)$ has order 
divisible by $r+1$.  Such an $E$ will exist, for example, if there is a
multiple of $(r+1)$ that is coprime to $q$ and that lies in the interval
$[q + 1 - 2\sqrt{q}, q + 1 + 2\sqrt{q}]$.  Since the group order of $E$ is 
divisible by $r+1$, there is a subgroup $G$ of $E(K)$ of order $r+1$.
Let $E'$ be the quotient of $E$ by this subgroup, so that there is an isogeny
$\varphi\colon E \to E'$ of degree $r+1$ having $G$ as its kernel.

The simplest way to construct an LRC from the morphism $\varphi$ is to 
take the $Q_i$ to be the points of $E'(K)$, other than the point at infinity,
that lie in the image of $E(K)$ under $\varphi$.  The points $P_{i,j}$ will 
then consist of all of the points of $E(K)$ that do not lie in the subgroup $G$,
and the helper sets will be the cosets of $G$ in $E(K)$ other than $G$ itself.

We can write $E$ and $E'$ in the form
\begin{align*}
E \colon \qquad y^2 + a_1 xy + a_3 y &= x^3 + a_2 x^2 + a_4 x + a_6\\
E'\colon \qquad v^2 + b_1 uv + b_3 v &= u^3 + b_2 u^2 + b_4 u + b_6
\end{align*}
and the embedding $\varphi^*\colon K(E')\to K(E)$ takes $u$ and~$v$ to
rational functions of $x$ and~$y$.

Let a positive integer $t$ be fixed. One choice for the functions
$f_1,\ldots,f_t \in K(E')$ is 
\[ f_1 = 1, \quad f_2 = u, \quad f_3 = v, \quad f_4 = u^2, \quad f_5 = uv, 
\quad f_6 = u^3, \quad f_7 = u^2v, \quad \ldots\]
and so forth, so that the $f_i$ are the functions on $E'$ with poles only at infinity, and 
of degree at most $t$.  Similarly, we can take the functions $e_1,\ldots,e_r\in K(E)$
to be 
\[ e_1 = 1, \quad e_2 = x, \quad e_3 = y, \quad e_4 = x^2, \quad e_5 = xy,
\quad e_6 = x^3, \quad e_7 = x^2y, \quad \ldots,\]
that is, the functions on $E$ with poles only at infinity and of degree at most~$r$.

We may take the divisor $D$ on $E$ to be
$D = tG + r\infty$; note that $D$ has degree $tr + t + r$.  The Goppa code we
obtain has the following parameters:
\begin{align*}
n &= \#E(K) - (r+1),\\
k &= tr,\\
d &= n - (tr + t + r),
\end{align*}
and if the matrices from Proposition~\ref{prop:matrices} meet the hypotheses of
that proposition, then our code will have locality~$r$.

Recall the Singleton-type bound on the minimum distance in \eqref{eq:singleton}. 
For our choices, we find that our minimum distance $d$ is $r+2$ less
than the Singleton bound.

\begin{example}
\label{ex:firstECGF64}
For our first example, we take $K = \FF_{64}$ and we take $E$ to be the elliptic curve $y^2 + y = x^3$ 
over $K$.  The group $E(K)$ has order $81$, so we may take $r = 2$ and
take $G$ to be the subgroup of order $3 = r + 1$ that contains the points
$(0,0)$ and $(0,1)$, along with $\infty$.
The quotient $E'$ of $E$ by $G$ can be written 
$v^2 + v = u^3 + 1$, with the isogeny $\varphi$ given by
\begin{align*}
u &= x + \frac{1}{x^2}\\
v &= y + \frac{1}{x^3}.
\end{align*}
(In fact, the curve $E'$ is isomorphic to $E$, but the equations for the isogeny 
are simpler if we use this alternate model for $E'$.)

Our functions $e_1$ and $e_2$ are $1$ and $x$, so the hypotheses of 
Proposition~\ref{prop:matrices} will be met for all nontrivial cosets of $G$
if $x(P_1) \ne x(P_2)$ for all $P_1$ and $P_2$, not in $G$,
that differ by a nonzero element of $G$. 
The $x$-coordinates of two distinct points on $E$ are equal if and
only if they are additive inverses, so we need to check that if $P$ is a 
point in $E(K)$ that does not lie in $G$, then $2P$ is not an element of $G$.
Since the group of points $E(K)$ has odd order, this is true.

We see that for any $t$ with $1\le t\le 25$, we get an LRC with length $78$,
dimension $2t$, and minimum distance $76 - 3t$.  For example, with $t=21$
we get a $(78,42,13)$-code with locality~$2$.
\end{example}

\begin{example}
\label{ex:firstECGF32}
In this example, we take $K = \FF_{32}$ and $E$ to be the elliptic curve
$y^2 + xy = x^3 + x$ over $K$.  There are $44$ points in $E(K)$, so we may take 
$r = 3$ and $G$ to be the cyclic subgroup of order $4$ consisting of the points 
on $E$ rational over $\FF_2$.  The quotient of $E$ by $G$ is a curve $E'$ that 
is isomorphic to $E$, so we may write $E'$ as $v^2 + uv = u^3 + u$,
where the isogeny $\varphi$ is given by
\begin{align*}
u &= \frac{(x^2 + x + 1)^2}{x(x+1)^2}\\
v &= \frac{(x^2 + x + 1)^2}{x^2(x+1)^2}\ y 
     + \frac{x^2 + x + 1}{x(x+1)^3}.
\end{align*}
We have $e_1=1$, $e_2 = x$, and $e_3 = y$, and we check by explicit computation
that the hypotheses of Proposition~\ref{prop:matrices} are met for 
all cosets of $G$ except for the trivial one.  We can therefore take
\begin{align*}
n &= 40\\
k &= 3t\\
d &= n -  (4t + 3).
\end{align*}
For a specific example, if we take $t = 7$ then we get a
$(40,21,9)$-code with locality~$3$.
\end{example}

\begin{example}
\label{ex:secondECGF32}
In this example, we again take $K = \FF_{32}$, we let $\alpha$ be an element 
of $K$ that satisfies $\alpha^5 + \alpha^2 + 1 = 0$, and we let $E$ be the
elliptic curve over $K$ defined by
\[ y^2 + x y = x^3 + x^2 + r^7 x.\]
The group $E(K)$ has order $42$, so we may take $r = 2$ and $G$ to be the 
cyclic subgroup of order $3$ consisting of the infinite point together with 
the points with $x$-coordinate equal to $\alpha^6$.
The quotient of $E$ by $G$ is a curve $E'$ that can be written
\[v^2 + u v = u^3 + u^2 + \alpha^{24} u + \alpha^6,\]
and the isogeny $\varphi$ is given by
\begin{align*}
u &= \frac{x (x + \alpha)^2}{(x + \alpha^6)^2}\\
v &= \frac{(x + \alpha)^2}{(x + \alpha^6)^2} y 
     + \frac{\alpha^6 x^2 + \alpha^{15} x + \alpha^{21}}{(x + \alpha^6)^3}.
\end{align*}
We have $e_1=1$ and $e_2 = x$, and we check by explicit computation
that the hypotheses of Proposition~\ref{prop:matrices} are met for 
all cosets of $G$ except the trivial one and the one that contains the unique 
nonzero $2$-torsion point of $E$.  To make this example work using the method
we have outlined above, we would therefore have to take our length $n$ to 
be~$36$, six less than the numbers of points in $E(K)$, which is smaller
than we might have hoped.  But there is another option, which we now explore.
\end{example}

In our general situation, 
let $L$ be the quadratic extension of $K$ and suppose there is a point 
$P$ be a point in $E(L)$ such that $P' := \varphi(P)$ does not lie in $E'(K)$.
Let $r' = \lceil r/2 \rceil$ and $t' = \lceil t/2\rceil$; then we can 
take $f_1,\ldots,f_t\in K(E')$ to be linearly independent elements of $L(t'P')$ and
$e_1,\ldots,e_r\in K(E)$ to be linearly independent elements of $L(r'P)$.
If we let $Q$ be the Galois conjugate of $P$, 
then we can take our divisor $D$ to be
\[D = r'(P + Q) + t' \sum_{R\in G} \left((P+R) + (Q + R)\right),\]
so that $D$ has degree $2r' + 2t'(r+1)$.

The point of taking functions with non-rational poles is that we can then
take the $Q_i$ to be all of the points of $E'(K)$ that lie in the image of $E(K)$
under $\varphi$, and take the points $P_{i,j}$ to be all of the points in~$E(K)$.
The Goppa code we obtain has parameters
\begin{align*}
n &= \#E(K),\\
k &= tr,\\
d &= n - (2r' + 2t'(r+1)).
\end{align*}
If $r$ and $t$ are both even, then $d = n - (tr + t + r).$
Again, if the matrices from Proposition~\ref{prop:matrices} meet the hypotheses of
that proposition, then our code will have locality~$r$.

\begin{example}
\label{ex:thirdECGF32}
Let us revisit the curves and maps from Example~\ref{ex:secondECGF32}.
We take $r=2$ and we take $G$ to be the same subgroup as before.
Let $\beta$ be a root of $x^2 + \alpha^9 x + \alpha^5$ in $L$,
and let $P$ be the point $(\beta, \beta^{626})$, so that
$P' = \varphi(P) = (\beta^{564}, \beta^{983})$.
We find that we can take $e_1 = 1$ and 
\[e_2 = \frac{y + \alpha^{26} x + \alpha^6}{x^2 + \alpha^9 x + \alpha^5}.\]
We check by explicit computation that the hypotheses of 
Proposition~\ref{prop:matrices} are met for all cosets of~$G$.
Therefore, for every even $t$ with $2\le t\le 12$
we obtain an LRC with locality $2$ and with
\begin{align*}
n &= 42,\\
k &= 2t,\\
d &= n - (3t + 2).
\end{align*}
For example, if we take $t=10$ we get a $(42,20,10)$-code with locality~$2$.
Our minimum distance is $4$ less than the Singleton bound.
\end{example}

\section{Locally recoverable codes from plane quartics}
\label{sec:quartics}

In this section, we construct locally-recoverable codes with locality $3$
by considering plane quartics whose automorphism groups contain a copy of
the Klein $4$-group~$V_4$.  Our analysis depends on whether or not the 
base field has characteristic~$2$. 

First let us consider the case where our base field $K = \Fq$ has odd
characteristic.  If $X$ is a nonsingular plane quartic over $K$ with
$V_4\subseteq \Aut X$, then $X$ is isomorphic to a curve defined by a
homogeneous quartic equation of the form $f(x^2,y^2,z^2) = 0$, where $f$ is a 
homogeneous quadratic.  One copy of $V_4$ in $\Aut X$
then contains the three commuting involutions given by 
$(x,y,z)\mapsto (-x,y,z)$ and $(x,y,z)\mapsto (x,-y,z)$
and $(x,y,z)\mapsto (x,y,-z)$.
Using the notation of Section~\ref{sec:general}, we can take $G$ to be this
copy of $V_4$ lying in $\Aut X$, and we take $Y$ to be the quotient curve $X/G$,
which is the genus-$0$ curve given by the homogeneous quadratic equation
$f(x,y,z) = 0$.  As usual, we let $\varphi\colon X\to Y$ be the canonical map.

Suppose we take $Q'$ to be any point of $Y(K)$ that is not in the image of
$X(K)$ under~$\varphi$.  For every $t > 0$ we let $f_1, \ldots, f_t$ be 
a basis for $\calL( (t-1) Q')$. 

For our three functions $e_1, e_2, e_3$ in $K(X)$
we take
\[
e_1 = 1, \qquad e_2 = x/z, \qquad e_3 = y/z,
\]
so that the $e_i$ all live in $\calL(D')$ for the degree-$4$ divisor formed
by intersecting $X$ with the line~$z=0$.  Then all of the 
products $e_i f_j$ are elements of $\calL(D)$, with $D = D' + 4\varphi^*(Q')$.
The divisor $D$ has degree $4t$.

We take the points $Q_1, \ldots, Q_s\in Y(K)$ to be the images of the
$P\in X(K)$ that have trivial stabilizers under the action of~$G$.  Each
$Q_i$ has, by definition, four preimages $P_{i,0}, P_{i,1}, P_{i,2}, P_{i,3}$ 
in~$X(K)$.  A point $P$ if $X(K)$ is one of the $P_{i,j}$ if and only if it does 
not lie on any of the three lines $x=0$, $y=0$, or $z=0$.

If $P\in X(K)$ is such a point, then $a := e_2(P)$ and $b := e_3(P)$ are nonzero elements of $K$,
and the matrix from Proposition~\ref{prop:matrices} is
\[
\begin{bmatrix*}[r]
1 &  a &  b \\
1 &  a & -b \\
1 & -a &  b \\
1 & -a & -b \\
\end{bmatrix*}.
\]
It is easy to check that all $3\times 3$ submatrices of this matrix have nonzero
determinant, so the hypotheses of the proposition are satisfied.  If $N$ is the 
number of points in $X(K)$ with $x(P), y(P), z(P)$ all nonzero, then for every
$t \le N/4 $ we have an LRC of locality $3$ with parameters
\begin{align*}
n&= N\\
k&= 3t\\
d&\ge n - 4t.
\end{align*}
The lower bound for $d$ is $2$ less than the generalized Singleton bound for LRCs.

\begin{example}
\label{ex:quartic7}
Consider the quartic curve $X$ over $\FF_{7}$ defined by
\[x^4 + y^4 + z^4 + 3 x^2 y^2 + 3 x^2 z^2 + 3 y^2 z^2 = 0.\]
There are $20$ points on $X$ with all coordinates nonzero.  We find that
for every $t\le 5$ we have a $(20, 3t, 20-4t)$-code with locality~$3$.
For example, with $t= 3$ we have a $(20,9,8)$-code with locality~$3$.
\end{example}

\begin{example}
\label{ex:quartic17}
Consider the quartic curve $X$ over $\FF_{17}$ defined by
\[x^4 + y^4 + 3 z^4 + 5 x^2 y^2 = 0.\]
There are $40$ points on $X$ with all coordinates nonzero.  We find that
for every $t\le 10$ we have a $(40, 3t, 40-4t)$-code with locality~$3$.
For example, with $t= 8$ we have a $(40,24,8)$-code with locality~$3$.
\end{example}

\begin{example}
\label{ex:quartic31}
Consider the quartic curve $X$ over $\FF_{31}$ defined by
\[x^4 + y^4 + z^4 + 4 x^2 z^2 + 7 x^2 y^2 = 0.\]
There are $60$ points on $X$ with all coordinates nonzero.  We find that
for every $t\le 15$ we have a $(60, 3t, 60-4t)$-code with locality~$3$.
For example, with $t= 13$ we have a $(60,39,8)$-code with locality~$3$.
\end{example}

Something completely analogous can be done over fields of characteristic~$2$.
Let $K=\Fq$ be a finite field, where $q$ is a power of~$2$.  
If $X$ is a nonsingular plane quartic over $K$ with
$V_4\subseteq \Aut X$, then $X$ is isomorphic to (the projective closure of)
a curve in the affine plane defined by a nonhomogeneous quartic equation 
of the form $f(x^2 + x,y^2+y) = 0$, where $f$ is a  bivariate quadratic.  
One copy of $V_4$ in $\Aut X$ then contains the three commuting involutions 
given by $(x,y)\mapsto (x+1,y)$ and $(x,y)\mapsto (x,y+1)$
and $(x,y)\mapsto (x+1,y+1)$.
Again we take $G$ to be this copy of $V_4$ lying in $\Aut X$, 
and we take $Y$ to be the genus-$0$ quotient curve $X/G$,
which has an affine model given by the equation $f(x,y) = 0$.  
Again we let $\varphi\colon X\to Y$ be the canonical map.

As in the odd characteristic case, 
we take $Q'$ to be any point of $K(Y)$ that is not in the image of
$X(K)$ under~$\varphi$.  For any $t > 0$ we let $f_1, \ldots, f_t$ be 
a basis for $\calL( (t-1) Q')$.  For our three functions in $K(X)$
we take $e_1 = 1$, $e_2 = x$, and $e_3 = y$. All of the 
products $e_i f_j$ are elements of $\calL(D)$, where $D = D' + 4\varphi^*(Q')$ and
where $D'$ is the divisor on $X$ obtained by intersecting $X$ with the line at infinity.
The divisor $D$ has degree $4t$.

We take the points $Q_1, \ldots, Q_s\in Y(K)$ to be the images of the
$P\in X(K)$ that have trivial stabilizers under the action of~$G$; these $P$ are
precisely the rational affine points of~$X$.  If $P = (a,b)$ is such a point, then
the matrix from Proposition~\ref{prop:matrices} is
\[
\begin{bmatrix*}[l]
1 & a   & b   \\
1 & a   & b+1 \\
1 & a+1 & b   \\
1 & a+1 & b+1 \\
\end{bmatrix*}.
\]
Once again, all $3\times 3$ submatrices of this matrix have nonzero
determinant, so the hypotheses of the proposition are satisfied.  If $N$ is the 
number of affine points in $X(K)$, then for every
$t \le N/4$ we have an LRC of locality $3$ with parameters
\begin{align*}
n&= N\\
k&= 3t\\
d&\ge n - 4t.
\end{align*}
The lower bound for $d$ is $2$ less than the generalized Singleton bound for LRCs.

\begin{example}
\label{ex:quartic8}
Consider the quartic curve $X$ over $\FF_{8}$ defined by
\[(x^2 + x)^2 + (x^2 + x) (y^2 + y) + (y^2 + y)^2 + 1 = 0.\]
There are $24$ affine points on~$X$.  We find that
for every $t\le 6$ we have a $(24, 3t, 24-4t)$-code with locality~$3$.
For example, with $t= 4$ we have a $(24,12,8)$-code with locality~$3$.
\end{example}

\begin{example}
\label{ex:quartic16}
Let $\alpha$ be an element of $\FF_{16}$ satisfying $\alpha^4 + \alpha + 1 = 0$.
Consider the quartic curve $X$ over $\FF_{16}$ defined by
\[(x^2 + x)^2 + (x^2 + x) (y^2 + y) + (y^2 + y)^2 + \alpha = 0.\]
There are $36$ affine points on~$X$.  We find that
for every $t\le 9$ we have a $(36, 3t, 36-4t)$-code with locality~$3$.
For example, with $t= 7$ we have a $(36,21,8)$-code with locality~$3$.
\end{example}

\begin{example}
\label{ex:quartic32}
Let $\alpha$ be an element of $\FF_{32}$ satisfying $\alpha^5 + \alpha^2 + 1 = 0$.
Consider the quartic curve $X$ over $\FF_{32}$ defined by
\[(x^2 + x)^2 + (x^2 + x) (y^2 + y) + \alpha^3 (y^2 + y)^2 + (y^2 + y) + \alpha^{26} = 0.\]
There are $64$ affine points on~$X$.  We find that
for every $t\le 16$ we have a $(64, 3t, 64-4t)$-code with locality~$3$.
For example, with $t= 14$ we have a $(64,42,8)$-code with locality~$3$.
\end{example}

The coverings $\varphi\colon X\to Y$ we have considered in this section
are biquadratic extensions of the projective line.  In the odd characteristic
case, this means that the function fields of our genus-$3$ curves $X$ are
obtained from $K(x)$ by adjoining the square roots of two separable polynomials
$f$ and $g$ of degree $3$ or~$4$ that have two roots in common (or one root in common,
if both $f$ and $g$ have degree~$3$).

Conversely, given two such functions $f$ and $g$, if we adjoin their square roots
to $K(x)$ we will either obtain the function field of a plane quartic as above,
or we will obtain the function field of a hyperelliptic curve of genus~$3$.  
(In slightly different terms, this situation was analyzed in~\cite[\S 4]{HoweLeprevostEtAl2000}.)
Generically, we obtain a plane quartic.  We close this section by considering the
non-generic case.

Using \cite[Prop.~14, p.~343]{HoweLeprevostEtAl2000}, one can show that over a
finite field of odd characteristic, if $X$ is a hyperelliptic curve of genus $3$
whose automorphism group contains a $V_4$-subgroup $G$ that does not contain the hyperelliptic
involution, then $X$ can be written in the form
\[
y^2 = a x^8 + b x^6 + c x^4 + b d^2 x^2 + a d^4,
\]
where the $V_4$-subgroup is generated by the involutions $(x,y)\mapsto(-x,y)$
and $(x,y) \mapsto (d/x, d^2 y/x^4)$.  We take $Y$ to be the quotient of $X$ by~$G$
and let $\varphi\colon X\to Y$ be the canonical map.

As in the plane quartic case, we can take $Q'$ to be any point of $Y(K)$ that
is not in the image of $X(K)$ under~$\varphi$.  For every $t > 0$ we let $f_1, \ldots, f_t$ be 
a basis for $\calL( (t-1) Q')$. 
For our three functions $e_1, e_2, e_3$ in $K(X)$ we take
$e_1 = 1$, $e_2 = x$, and $e_3 = x^2$, 
so that the $e_i$ all live in $\calL(D')$ where $D'$ is twice the (degree-$2$) divisor
at infinity.  Then all of the products $e_i f_j$ are elements of $\calL(D)$,
with $D = D' + 4\varphi^*(Q')$. The divisor $D$ has degree $4t$.

As in the plane quartic case, we take the points $Q_1, \ldots, Q_s\in Y(K)$
to be the images of the $P\in X(K)$ that have trivial stabilizers under the
action of~$G$, which means the points $P$ with finite nonzero $x$-coordinates whose
squares are neither $d$ nor~$-d$. Each $Q_i$ has
four preimages $P_{i,0}, P_{i,1}, P_{i,2}, P_{i,3}$ in~$X(K)$, and the $x$-coordinates of
the $P_{i,j}$ are distinct.

For a given helper set $H_i = \{P_{i,0}, P_{i,1}, P_{i,2}, P_{i,3}\}$ we find that the
matrix from Proposition~\ref{prop:matrices} is
\[
\begin{bmatrix}
1 &  x_0 &  x_0^2 \\
1 &  x_1 &  x_1^2 \\
1 &  x_2 &  x_2^2 \\
1 &  x_3 &  x_3^2 \\
\end{bmatrix},
\]
where $x_j$ is the $x$-coordinate of $P_{i,j}$.
The $3\times 3$ submatrices are all Vandermonde matrices, and have nonzero
determinant because the $x_j$ are distinct; thus we can construct an LRC from
this setup.
If $N$ is the number of points in $X(K)$ with $x(P)$ finite and nonzero
and with $x(P)^2 \ne \pm d$,
then for every $t \le N/4 $ we have an LRC of locality $3$ with parameters
\begin{align*}
n&= N\\
k&= 3t\\
d&\ge n - 4t.
\end{align*}
The lower bound for $d$ is $2$ less than the generalized Singleton bound for LRCs.

\begin{example}
\label{ex:hyperelliptic31}
Consider the hyperelliptic curve $X$ over $\FF_{31}$ defined by
\[y^2 = x^8 + 16 x^6 + 14 x^4 + 16 x^2 + 1.\]
There are $56$ points on $X$ whose $x$-coordinates are finite, nonzero, and whose
squares are not~$\pm 1$.  We find that
for every $t\le 14$ we have a $(56, 3t, 56-4t)$-code with locality~$3$.
For example, with $t= 12$ we have a $(56,36,8)$-code with locality~$3$.
This is worse than Example~\ref{ex:quartic31}, which uses a plane quartic over
$\FF_{31}$.  
\end{example}

Indeed, for most~$q$ we expect to get better results from plane quartics
than from hyperelliptic genus-$3$ curves, simply because there are more plane
quartics with $V_4$ actions than there are hyperelliptic genus-$3$ curves with
(nonhyperelliptic) $V_4$ actions;
this makes it likely that the genus-$3$ curve with a $V_4$ action
having the largest number of points will be a plane quartic.

\section{Locally recoverable codes from higher genus curves}
\label{sec:highergenus}

The possibility of using automorphism of curves of relatively high genus
in order to obtain LRCs was already used in~\cite{BargTamoEtAl2015}.  In this 
section we consider some examples of constructions of curves with automorphisms
of order~$3$, in order to obtain LRCs of locality~$2$.

Our analysis is simplest over finite fields $K = \Fq$ with $q\equiv 1\bmod 3$, because
then our degree-$3$ Galois extension $\varphi\colon X\to Y$ can be written as a Kummer extension.
Suppose $Y$ is a curve of genus $g$ over such a field $K$, and let $h$ be an element of 
$K(Y)$ such that the function field $K(X)$ is obtained from $K(Y)$ by adjoining an
element $z$ such that $z^3 = h$.  We can take our points $Q_i$ in $Y(K)$ to be the
points $Q$ such that $h(Q)$ is a nonzero cube in~$K$.  

Suppose we take our functions $f_1,\ldots,f_t$ to be a basis for $\calL(D')$ for
some divisor $D'$ on $Y$, and suppose we take our functions $e_1$ and $e_2$ 
to be $e_1 = 1$ and $e_2 = z$.  Note that the degree of $z$ (as a function on $X$)
is equal to the degree of $h$ (as a function on $Y$).  If we take $D$
to be the pullback $\varphi^*(D')$ of $D'$ to $X$, plus the polar divisor of $z$,
then each $e_i f_j$ lies in $\calL(D)$, and the degree of $D$ is $3$ times the
degree of $D'$ plus the degree of~$h$.

The parameters of the code we obtain are then:
\begin{align*}
n&= 2s\\
k&= 2t\\
d&\ge n - 3\deg D' - \deg h.
\end{align*}

In this section we work through several examples of this general construction.
We end with some examples where the base field does \emph{not} contain the
cube roots of unity.

\begin{example}
\label{ex:hermitian16}
Let $K = \FF_{16}$ and let $X$ be the Hermitian curve $y^4 + y = x^5$ over~$K$.
In~\cite{BargTamoEtAl2015}, the projections of this curve to the projective line
via the $x$- and $y$-coordinates were used to create LRCs of locality $3$ and~$4$.
Here we consider the group $G$ of automorphisms of $C$ generated 
by the automorphism of order $3$ given by
\[(x,y) \mapsto (\zeta x,\zeta^2 y),\]
where $\zeta\in K$ is a primitive cube root of unity.

The quotient $Y$ of $X$ by $G$ is a genus-$2$ curve that can be written
$z^2 + z = w^5$, with the cover $\varphi\colon X\to Y$ being given by 
$w = y/x^2$, $z = y/x^5$.  Note that then $y^3 = (z+1)/z$, and the cubic extension 
of function fields $\varphi^*\colon K(Y)\to K(X)$ is given by adjoining $y$ to
$K(Y)$.

Let $R_0$ be the point $(w,z) = (0,0)$ of $Y$ and let $R_1$ be the point $(0,1)$ of~$Y$.
We compute that for every integer $s \ge 3$ we have
\begin{align*}
\calL(sR_0 + sR_1)     &= \left\{1, \frac{1}{w}, \frac{1}{w^2}, \frac{z}{w^3}, \frac{1}{w^3}, \frac{z}{w^4}, \frac{1}{w^4}, \ldots, \frac{z}{w^s}, \frac{1}{w^s}\right\}\\
\intertext{and}
\calL(sR_0 + (s+1)R_1) &= \left\{1, \frac{1}{w}, \frac{1}{w^2}, \frac{z}{w^3}, \frac{1}{w^3}, \frac{z}{w^4}, \frac{1}{w^4}, \ldots, \frac{z}{w^s}, \frac{1}{w^s}, \frac{z}{w^{s+1}}\right\}.
\end{align*}
This shows that for any $t\ge 3$, if we take $f_1,\ldots,f_t$ to be the functions
\[f_1 = 1, \quad f_2 = \frac{1}{w}, \quad f_3 = \frac{1}{w^2}, \quad f_4 = \frac{z}{w^3}, \quad f_5 = \frac{1}{w^3}, \quad f_6 = \frac{z}{w^4}, \quad f_7 = \frac{1}{w^4}, \quad \ldots, \]
then $f_1,\ldots,f_t$ lie in $\calL(D')$ for a divisor $D'$ on $Y$ of degree $t+1$.

The points $R_0$ and $R_1$ are the only two points that ramify in the cover 
$\varphi\colon X\to Y$, and their preimages in $X$ are $\infty$ and $(0,0)$, 
respectively.  That leaves $63$ other points in $X(K)$ that are not fixed by the
group $G$, and these points map down to $21$ points on $Y$.  We take $Q_1,\ldots,Q_{21}$
to be those points, and for every $i$ we let $P_{i,1}$, $P_{i,2}$, and $P_{i,3}$ be
the three points lying over $Q_i$.  Note that none of the $P_{i,j}$ have $y$-coordinate equal
to~$0$, because the only point on $X$ with $y$-coordinate $0$ is the point $(0,0)$
that maps down to~$R_1$.

We take the functions $e_1$ and $e_2$ in $K(X)$ to be $e_1 = 1$ and $e_2 = y$.
For each $i$, if $P_{i,1} = (x_1,y_1)$, then $P_{i,2}$ and $P_{i,3}$ are
$(\zeta x_1, \zeta^2 y_1)$ and $(\zeta^2 x_1, \zeta y_1)$, and the matrix from
Proposition~\ref{prop:matrices} is
\[
\begin{bmatrix}
1 & y_1\\
1 & \zeta^2 y_1\\
1 & \zeta y_1\\
\end{bmatrix}.
\]
Since $y_1\ne 0$, each $2\times 2$ submatrix of this matrix is invertible, so
the code we construct as in Section~\ref{sec:general} has locality~$2$.

If the functions $f_1,\ldots,f_t$ in $K(Y)$ lie in $\calL(D')$ for a degree-$(t+1)$ divisor~$D$,
then the functions $e_i f_j$ in $K(X)$ all lie in the Riemann--Roch space of $\varphi^*(D) + 5\infty$,
a divisor of degree $3t + 8$.
We find that for every integer $t$ with $ 3\le t\le 18$, we obtain a Goppa
code with paramaters
\begin{align*}
n &= 63\\
k &= 2t\\
d &\ge n - (3t + 8).
\end{align*}
The Singleton bound on $d$ is $n - 3t + 2$, so we are $10$ below the Singleton bound.

By taking $t=16$ we obtain a $(63,32,7)$-code with locality $2$; by taking 
$t = 14$ we get a $(63,28,13)$-code with locality $2$. 
\end{example}

Over a base field containing primitive cube roots of unity,
we can produce examples of $\varphi\colon X\to Y$ that are Galois covers of an
elliptic curve $Y$, with $X$ of genus as large as~$7$, as follows:


We pick an arbitrary elliptic curve $Y$ over $K = \FF_q$ (with $q\equiv 1 \bmod 3$)
and consider functions $g$ that are ratios of functions of the form $ay + bx + c$.
Suppose neither the numerator nor the denominator of this representation of $g$ 
is a constant times a cube in $K(Y)$.
Let $X$ be the curve whose function field is obtained from $K(Y)$ by adjoining a cube
root $z$ of $g$, and let $\varphi\colon X\to Y$ be the resulting cover.  
Since $\varphi$ is ramified at exactly the zeros and poles of $g$ whose orders are not 
multiples of~$3$, the Riemann--Hurwitz formula tells us that the genus of $X$ is at most~$7$.

We can take our points $Q_1,\ldots,Q_s$ to be the points $Q$ of $Y$ such that $f(Q)$ is
a nonzero cube; such points split completely in the cover $\varphi$.  If $Y$ is given
by a Weierstrass equation 
\[y^2 + a_1 xy + a_3 y = x^3 + a_2 x^2 + a_4 x + a_6\]
then we can take our functions $f_1,\ldots, f_t$ to be the usual basis for
$\calL(t\infty)$, namely
\[ f_1 = 1, \quad f_2 =x, \quad f_3 = y, \quad f_4 = x^2, \quad f_5 = xy, \quad \ldots.\]
We take $e_1 = 1$ and $e_2 = z$.  Note that $e_2(P)\ne 0$ for every point $P$ lying over 
one of our $Q_i$, so that we will indeed obtain a code of locality $2$.  Also, if we
let $D'$ be the degree-$3$ divisor $\varphi^*(\infty)$, then every $e_i f_j$ lies in 
$\calL((t+1) D')$.  Using these functions and the points $P$ lying over the $Q_i$, we 
obtain an LRC with locality $2$ having parameters
\begin{align*}
n &= 3s,\\
k &= 2t,\\
d &\ge n - 3t - 3.
\end{align*}
The Singleton bound for $d$ is $n - 3t + 2$, so we are at worst $5$ away from this bound.

\begin{example}
\label{ex:triplecoverF16}
Let $K = \FF_{16}$ and let $\alpha\in K$ satisfy $\alpha^4 + \alpha + 1 = 0$.
We take $Y$ to be the elliptic curve $y^2 + \alpha y = x^3$
and we take $g$ to be the function 
\[g = \frac{y + \alpha^4 x + \alpha^3}{y + \alpha^{10} x + \alpha^3}.\]
The resulting curve $X$ has genus~$7$.
Of the $21$ points $Q$ on $Y$, there are $15$ for which $g(Q)$ is a nonzero cube.
Thus we may take $s=15$.  We find that for every $t$ with $1\le t \le 13$, we obtain
a $(45, 2t, 42 - 3t)$-code with locality $2$.  For example, with $t=11$ we get 
a $(45, 22, 9)$-code with locality $2$.
\end{example}

\begin{example}
\label{ex:triplecoverF64}
Let $K = \FF_{64}$, let $Y$ be the elliptic curve $y^2 + y = x^3 + 1$, and let 
and let $g = y/x$.  The curve $X$ has genus~$7$. Of the $81$ rational points $Q$ on $Y$, there are
$57$ such that $g(Q)$ is a nonzero cube.  Thus we can take $s = 57$, and we find that 
every $t$ with $1\le t \le 55$, we obtain
a $(171, 2t, 168 - 3t)$-code with locality $2$.  For example, with $t=51$ we get 
a $(171, 102, 13)$-code with locality $2$.
\end{example}

Over fields $K = \Fq$ that do not contain the cube roots of unity we can no longer
use Kummer theory to write down cyclic cubic extensions, but there  is still a normal
form for cubic Galois extensions.   Assume that the characteristic of $K$ is not $3$ and that 
$K$ does not contain the cube roots of unity, and let $Y$ be a curve over $K$.  Every
cubic Galois extension of the function field $K(Y)$ can be obtained by adjoining
a root of a polynomial of the form
\begin{equation}
\label{EQ:GaloisCubic}
z^3 - 3f z^2 - 3(f+1)z - 1,
\end{equation}
where $f$ is function on $Y$.  The ramification points of the
resulting cover $\varphi\colon X\to Y$ consist of the geometric points $P$
at which the function $f^2 + f + 1$ has a zero whose order is not a multiple of~$3$.
If $a$ is a function on $Y$ other than the constant functions $0$ and $-1$, we can set
$b = (a^3 - 3a - 1)/(3a^2 + 3a)$ and take $g = (bf - 1)/(f + b + 1)$.  Then adjoining
a root of 
\[
z^3 - 3g z^2 - 3(g+1)z - 1
\]
to $K(Y)$ will give a function field isomorphic to $K(X)$, and all functions $g$
that give extensions of $K(Y)$ isomorphic to $K(X)$ arise in this way.

If we let $w\in K(X)$ be a root of the polynomial in Equation~\eqref{EQ:GaloisCubic},
then
\[ 3 f =  w+ \frac{-1}{w+1} + \frac{-w-1}{w},\]
and from this we see (by looking, for example, at the poles of the right hand side)
that the degree of $f$ as a function on $X$ is three times the degree of $w$. 
It follows that the degree of $w$ is equal to the degree of $f$ as a function on~$Y$.

\begin{example}
\label{ex:triplecoverF32}
Let $K = \FF_{32}$ and let $\alpha\in K$ satisfy $\alpha^5 + \alpha^2 + 1 = 0$.
We take $Y$ to be the elliptic curve $y^2 + xy = x^3 + 1$, and we take $f$ to be
the degree-$3$ function
\[
f = \alpha^2 \frac{y + \alpha^3 x}{y + \alpha^2 x}.
\]
This gives rise to a genus-$7$ covering curve~$X$.
We check that $29$ points $Q$ of $Y$ split in the resulting extension
$\varphi\colon X\to Y$, so we can take up to $29$ points $Q_i$.
We take the functions $f_i$ to be the usual ones when $Y$ is an elliptic
curve:
\[ f_1 = 1, \quad f_2 =x, \quad f_3 = y, \quad f_4 = x^2, \quad f_5 = xy, \quad \ldots,\]
so that $\{f_1,\ldots,f_t\}$ form a basis for $\calL(t\cdot\infty).$
We take $e_1 = 1$ and $e_2 = w$, where $w$ satisfies the polynomial~\eqref{EQ:GaloisCubic}.
Note that the degree of $e_2$ is~$3$.
The value of $w$ on the three points of $X$ lying over a splitting point 
of $Y$ are distinct, so the hypotheses of Proposition~\ref{prop:matrices} 
are met.  We find that for every $t\le 28$ we can construct a code of locality $2$ with
\begin{align*}
n &= 87,\\
k &= 2t,\\
d &\ge n - 3t - 3.
\end{align*}
Our minimum distance is $5$ less than the Singleton bound for LRCs.

Compare this example to Example~\ref{ex:thirdECGF32}.  We see that we are 
$1$ farther away from the Singleton bound, but the dimension of the code can be 
taken to be much larger.
\end{example}

\section{The availability problem}
\label{sec:availability}

Recall that an $(n,k,d)$-code has \emph{locality} $r$ if for 
every index $i\in\{1,\ldots,n\}$ there is a recovery set of size at most $r$ such that the coordinate $i$
in every codeword can be determined from the coordinates in the set $I_i;$ see Definition~\ref{def:LRC}.

In some circumstances, it is desirable to have \emph{more that one} such
recovering set $I_i$ for each $i$.  The problem of constructing codes with
multiple recovering sets is called the \emph{availability problem}, because
such codes make it possible for multiple users to recover lost coordinates
with less impact on bandwidth usage.

\begin{definition}[LRC codes with availability]\label{def:LRC-t} 
A code $\calC\subset \Fq^n$ of size $q^k$ is said to have $t$ recovery sets if
for every coordinate $i\in\{1,\ldots,n\}$ and every $x\in \calC$ condition \eqref{eq:def1}
holds true for pairwise disjoint subsets $I_{i,j}\subset\{1,\ldots,n\}\backslash\{i\},|I_{i,j}|=r_j,j=1,\dots,t.$ 
\end{definition}

We will focus on the case where we have \emph{two} recovering sets $I_{i,1}$
and $I_{i,2}$ for each index $i$, where we assume further that for each $i$
the sets $I_{i,1}$ and $I'_{i,2}$ are disjoint, and have cardinalities $r_1$ and~$r_2$.
In this section we will show how to construct LRCs with dual recovering sets
from elliptic curves.  Some of the choices we make will be for convenience of
exposition; alternate choices could produce codes with better parameters.

Given integers $r_1>1$ and $r_2>1$, choose a prime power $q$ such that there is an 
elliptic curve $E$ over the finite field $K = \Fq$ such that $E(K)$ has 
two subgroups $G_1$ and $G_2$, of order $r_1+1$ and $r_2+1$, respectively, such that
$G_1\cap G_2 = \{0\}$.  Let $G$ be the subgroup of order $(r+1)(r'+1)$ generated
by $G_1$ and $G_2$, let $E_1$ and $E_2$ be the quotients of $E$ by $G_1$ and $G_2$,
respectively, and let $E'$ be the quotient of $E$ by $G$.  Then we have a 
diagram
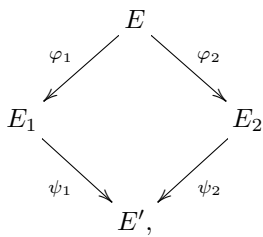
\[
\xymatrix{
                      & E\ar_{\varphi_1}[dl]\ar^{\varphi_2}[dr] &                     \\
E_1\ar_{\psi_1}[dr]   &                                         & E_2\ar^{\psi_2}[dl] \\
                      &                    E',                  &                     \\
}
\]
where $\varphi_1$, $\varphi_2$, $\psi_1$, and $\psi_2$ are the natural isogenies.
We let $\phi := \psi_1\circ\varphi_1 = \psi_2\circ\varphi_2.$

Pick points $Q'\in E'(K)\setminus \varphi(E(K))$ and 
$Q_1\in E_1(K)\setminus \varphi_1(E(K))$ and
$Q_2\in E_2(K)\setminus \varphi_2(E(K))$.  Choose a basis
$\{e_{1,1}, e_{1,2}, \ldots, e_{1,r_2}\}$
for $\calL(r_2 Q_1)$ and a basis
$\{e_{2,1}, e_{2,2}, \ldots, e_{2,r_1}\}$
for $\calL(r_1 Q_2)$.  Given an integer $t>1$, choose a basis
$\{f_1, f_2, \ldots, f_t\}$
for $\calL(t Q')$.

Define divisors $D_1, D_2,$ and $D'$ on $E$ by setting $D_1 = \varphi_1^{-1}(Q_1)$,
$D_2 = \varphi_2^{-1}(Q_2)$, and $D' = \varphi^{-1}(Q')$, so that $D_1$ has degree $r_1 + 1$,
$D_2$ has degree $r_2 + 1$, and $D'$ has degree $(r_1+1)(r_2+1)$.

From this data, we construct a Goppa code as follows.
Let $P_1, \ldots, P_n$ be the points in $E(K)$.  Let $S$ be the set of all triples
$(h,i,j)$ of integers with $1\le h\le r_1$ and $1\le i\le r_2$ and $1\le j\le t$,
and let $T$ be the set of integers $i$ with $1\le i\le n$.
For every vector $\aa = (a_{h,i,j}) \in K^S$, 
let $f_\aa$ be the function
\[f_\aa := \sum_{h=1}^{r_1} \sum_{i=1}^{r_2} 
         \varphi_1^* e_{1,h} \ \varphi_2^* e_{2,i} \sum_{j=1}^t a_{h,i,j} \ \varphi^* f_j,\]
and set
$\gamma(\aa) = \bb = (b_i)\in K^T$ where
$b_{i} = f_\aa(P_{i})$ for all $i=1,\ldots,n$. 
Note that each function $f_\aa$ lies in $\calL(D)$, where $D = r_2 D_1 + r_1 D_2 + tD'$. 
It follows that the code we have constructed as parameters
\begin{align*}
n &= \#E(K),\\
k &= r_1 r_2 t,\\
d &= n - (r_1+1)(r_2+1)t - r_2 (r_1+1) - r_1(r_2+1) \\
  & = n - (r_1+1)(r_2+1)t - 2 r_1 r_2  - r_1 - r_2.
\end{align*}

Each point $P_i$ lies in two helper sets: its orbit $H_{i,1}$ under $G_1$,
and its orbit $H_{i,2}$ under~$G_2$.  This gives us two ways to view the
situation as as example of the construction in Section~\ref{sec:ECs}. 
We will explain this for the group $G_1$, the other choice being 
completely analogous.

There are $r_2 t$ functions of the form $e_{1,i} \,\psi_1^* f_j$ on the elliptic
curve $E_1$; let us label these function $\zeta_1,\ldots, \zeta_\tau$, where
$\tau = r_2 t$.  Let us also write $\varepsilon_i = \varphi_2^* e_{2,i}$ for
$i = 1,\ldots, r_1$.  Then, as in Section~\ref{sec:ECs}, we have a covering
$E \to E_1$ of degree $r_1 + 1$, together with $r_1$ functions $\varepsilon_i$ on $E$
and $\tau$ functions $\zeta_1,\ldots,\zeta_\tau$ on $E_1$.  The helper sets
$H_{i,1}$ will give us an LRC with locality $r_1$ provided the hypotheses
of Proposition~\ref{prop:matrices} hold.  In general these hypotheses will 
have to be checked for the particular choices made in the construction.
In the case where $r_1 = r_2 = 2$, these hypotheses are always satisfied.

\begin{theorem}
\label{T:availability2}
If, in the above construction, we have $r_1 = r_2 = 2$,
then the code we obtain is an LRC with $2$ recovery sets.
Its parameters satisfy
\begin{align*}
n &= \#E(K),\\
k &= 4 t,\\
d &= n - 9t - 12.
\end{align*}
\end{theorem}

\begin{proof}
When $r_1 = r_2 = 2$, we have two order-$3$ subgroups $G_1$ and $G_2$ (intersecting only
at the identity) and two degree-$3$ maps $\varphi_1\colon E \to E_1$ and $\varphi_2\colon E\to E_2$.
With respect to the group $G_1$, the situation as described in the paragraph preceding the
statement of the theorem is as follows:

There are $2 t$ functions $\zeta_1,\ldots, \zeta_\tau$ on the elliptic curve $E_1$,
where $\tau = 2 t$, and we also have two functions $\varepsilon_1 = \varphi_2^* e_{2,1}$
and $\varepsilon_2 = \varphi_2^* e_{2,2}$ on $E$.  Without loss of generality
we may specify that the basis element $e_{2,1}$ of $\calL(2 Q_1$ is equal to $1$, 
so $\varepsilon_1=1$ as well.  To check the hypotheses of Proposition~\ref{prop:matrices}
we must ask: For every point $P$ of $E(K)$ and every nonzero element $Q$ in $G_1$,
is the matrix
\[
\begin{bmatrix}
1 & \varepsilon_2(P) \\
1 & \varepsilon_2(P+Q) \\
\end{bmatrix}
\]
invertible?  

Let $\alpha = \varepsilon_2(P) = e_{2,2}(\varphi_2(P))$.  The function $e_{2,2}$ on $E_2$
has degree $2$ has a double pole at $Q_2$ and no other poles, so the divisor of 
$e_{2,2} - \alpha$ is equal to $(R_1) + (R_2) - 2(Q_2)$ for some points $R_1$ and $R_2$ of~$E_2$.
One of these points, say $R_1$, must be $\varphi_2(P)$.  Furthermore, in the group $E_2(K)$
we must have $R_1 + R_2 = 2 Q_2$.

If the matrix above were not invertible then we would have
$\varepsilon_2(P+Q) = \alpha$, which would imply that $\varphi_2(P+Q)$ is either $R_1$ or $R_2$.
We cannot have $\varphi_2(P+Q) = R_1 = \varphi_2(P)$, because that could only happen
if $Q$ were in the kernel $G_2$ of $\varphi_2$, and we know that $G_1$ and $G_2$ have
trivial intersection.  On the other hand, suppose $\varphi_2(P+Q) = R_2$.  Then
\[
\varphi_2(P) + \varphi_2(P+Q) = 2 Q_2,
\]
and since $Q$ is a $3$-torsion point we have $2\varphi(P + 2Q) = 2 Q_2$.  This means
that $Q_2$ differs from $\varphi(P + 2Q)$ by a $2$-torsion point $T$.  But every $2$-torsion
point in $E_2(K)$ lies in  $\varphi_2(E(K))$ because $\varphi_2$ has degree~$3$, so 
$Q_2$ must lie in the image of $E(K)$ under $\varphi_2$, contrary to how it was chosen.
Therefore, the code we constructed is an LRC, with helper
sets equal to the cosets of~$G_1$.

The same argument shows that our code is an LRC with helper sets equal to the cosets of~$G_2$,
so we have constructed an LRC with $2$ recovery sets.  The parameters of the code were
calculated in the discussion before the atatement of the theorem.
\end{proof}

If $E$ is an elliptic curve over $K$ such that $E(K)$ 
has two order-$3$ subgroups $G_1$ and $G_2$ that intersect only 
at the identity, then all of the $3$-torsion points of $E$ are
rational over~$K$, and the subgroup of $E(K)$ generated by $G_1$ and $G_2$
is $E[3](K)$.  Note that then the curve $E'$ is isomorphic to $E$, and the isomorphism
can be chosen so that $\psi_2\circ \varphi_2 = \psi_1\circ \varphi_1 = 3$.

Note also if $E/K$ has all of its $3$-torsion defined over $K$, then the
Galois equivariance of the Weil pairing shows that $K$ must contain the 
cube roots of $1$, so that $q \equiv 1 \bmod 3.$

\begin{example}
\label{ex:availability64}
We construct an LRC over $K = \FF_{64}$ with $2$ recovery sets, each of size~$2$,
by using the above construction with $r_1 = r_2 = 2$.

Let $E$ be the elliptic curve $y^2 + y = x^3$ over~$K$.
We take $E_1 = E_2 = E' = E$, and we define commuting isogenies $\varphi_1$, $\varphi_2$, $\psi_1$,
and $\psi_2$ of degree $3$ by
\begin{align*}
\varphi_1(x,y) &= \left(\frac{x^3 + x^2 + x}{(x + 1)^2}, y + \frac{x^4 + x^3 + x^2 + x + 1}{(x + 1)^3}\right),\\
\varphi_2(x,y) &= \left(\frac{x^3 + x^2 + 1}{x^2}, y + \frac{x^4 + x + 1}{x^3}\right),\\
\end{align*}
and by taking $\psi_1 = \varphi_2$ and $\psi_2 = \varphi_1$.  
Then $\psi_1\circ\varphi_1= \psi_2\circ\varphi_2$
is equal to the multiplication-by-3 map on~$E$.

The kernel of $\varphi_1$ consists of the identity element and the two points 
in $E(K)$ with $x$-coordinate equal to $1$;
the kernel of $\varphi_2$ consists of the identity element and the two points
in $E(K)$ with $x$-coordinate equal to $0$.

Theorem~\ref{T:availability2} shows that the  construction from this section gives 
us an LRC with $2$ recovery sets, with parameters
$n = 81$, $k = 4 t$, and $d = 69 - 9t.$
For example, with $t=7$, we get an $(81,28,6)$-code over $\FF_{64}$ with two
recovery sets.
\end{example}


\section{Locally recoverable codes from algebraic surfaces}
\label{sec:higherdimension}

In this section we use the ideas from section~\ref{sec:general} to construct the first examples of locally recoverable codes using algebraic \emph{surfaces}.  Notably, in Example~\ref{ex:q^2+2} we construct a $(18,11,3)$-code with locality $2$ code that meets the Singleton-type bound~\eqref{eq:singleton}. This code has distance $d = 3$ while having large length with respect to the size of its alphabet: $n = q^2 + 2$. This example shows that algebraic surfaces can be used to produce optimal LRC codes of large length.  We also use a K3 surface in Example~\ref{ex:K3} to construct a code of locality $3$ with parameters $(24,17,3)$; this code also meets the Singleton-type bound and satisfies $n = q^2 - 1$.

\subsection{General construction}

We use smooth surfaces in $\PP^3$ over $K = \FF_q$ of the form
\[
X \colon \qquad w^{r+1} = f_{r+1}(x,y,z),
\]
where $f_{r+1}(x,y,z)$ is a homogeneous polynomial in $x$, $y$, and $z$ of degree $r+1$.  The projection map $\PP^3 \dasharrow \PP^2$ sending $[x,y,z,w] \mapsto [x,y,z]$ restricts to a morphism $\varphi\colon X \to \PP^2$.  If $r+1 \mid q-1$ then the nonempty fibers of $\varphi$ above $K$-points in $\PP^2$ outside the branch locus
\[
C \colon \qquad f_{r+1}(x,y,z) = 0
\]
consist of $r+1$ distinct points, because if the equation $w^{r+1} = \alpha \neq 0$ has a root, then it has $r+1$ distinct roots in $K^\times$.  Taking the fibers of $\varphi$ as helper sets, we construct a code with locality $r$.  Let $U \subseteq \PP^3$ denote the open affine subset $z\neq 0$; we shall use the $K$-points of $(X\setminus \varphi^{-1}(C)) \cap U$ as inputs for the evaluation code; hence, the length of the code will be 
\[
n = \#\left( (X\setminus \varphi^{-1}(C)) \cap U\right)(K).
\]
The Picard group of the base variety $\PP^2$ is isomorphic to $\ZZ$, and every effective divisor is linearly equivalent to $mL$ for some $m \geq 0$, where $L$ is a line on $\PP^2$.  We use sections contained in a vector space of the form $H^0(\PP^2,\scrO_{\PP^2}(mL))$ to construct functions $f$ as in section~\ref{sec:general}. This vector space can be identified with homogeneous polynomials of degree $m$ in $x$, $y$ and $z$.  Let $S(m)$ be the set of all triples $(i,j,l)$ of nonnegative integers with $i + j + l = m$.  For $(i,j,l) \in S(m)$, we let $f_{i,j,l} = x^iy^jz^l/z^m \in K(\PP^2)$.

Fix a positive integer $m$. For $i = 1,\dots,r$, let $e_i = (w/z)^{i - 1}$ be elements of $K(X)$. For $\aa = (a_{i,j,l}) \in K^{S(m)}, \bb = (b_{i,j,l}) \in K^{S(m-1)},\dots, \cc = (c_{i,j,l}) \in K^{S(m-r+1)}$, we let $f_{\aa,\bb,\dots,\cc} \in K(X)$ be the function
\[
\begin{split}
f_{\aa,\bb,\dots,\cc} 	&= e_1\cdot\Bigg(\sum_{(i,j,l) \in S(m)} a_{i,j,l}f_{i,j,l} \Bigg) + e_2\cdot\Bigg(\sum_{(i,j,l) \in S(m-1)} b_{i,j,l}f_{i,j,l} \Bigg) \\
				&\quad + \cdots + e_r\cdot\Bigg(\sum_{(i,j,l) \in S(m-r+1)} c_{i,j,l}f_{i,j,l} \Bigg) \\
			&= \frac{1}{z^m}\cdot\Bigg(\sum_{(i,j,l) \in S(m)} a_{i,j,l}x^iy^jz^l \Bigg) + \frac{w}{z^{m+1}}\cdot\Bigg(\sum_{(i,j,l) \in S(m-1)} b_{i,j,l}x^iy^jz^l \Bigg) \\
			&\quad + \cdots + \frac{w^{r-1}}{z^{m+1}}\cdot\Bigg(\sum_{(i,j,l) \in S(m-r+1)} c_{i,j,l}x^iy^jz^l \Bigg)
\end{split}
\]
and we define a map 
\[
\gamma\colon K^{S(m)} \times K^{S(m-1)}\times\cdots\times K^{S(m-r+1)} \to K^n\qquad \gamma(\aa,\bb,\dots,\cc) \mapsto (f_{\aa,\bb,\dots,\cc}(P)),
\]
 where $P \in \PP^3(K)$ ranges over the $K$-points of $(X\setminus \varphi^{-1}(C)) \cap U$.

Let $P_0 = [x_0,y_0,z_0,w_0]$ be a $K$-point in $(X\setminus \varphi^{-1}(C)) \cap U$, and let $\zeta \in K^\times$ be a primitive $(r+1)$-th root of unity.  The fiber of the map $\varphi$ above $Q = [x_0,y_0,z_0]$ consists of $P_0$ and the points 
\[
P_1 = [x_0,y_0,z_0,\zeta w_0], P_2 = [x_0,y_0,z_0,\zeta^2 w_0], \dots, P_r = [x_0,y_0,z_0,\zeta^r w_0].
\]
The matrix for the analog of Proposition~\ref{prop:matrices} in this setting is
\[
\begin{pmatrix}
1 & w_0/z_0 & \cdots & (w_0/z_0)^{r-1} \\
1 & \zeta w_0/z_0 & \cdots & \zeta^{r-1}(w_0/z_0)^{r-1}\\
 & & \vdots & \\
1 & \zeta^r w_0/z_0 & \cdots & \zeta^{r(r-1)}(w_0/z_0)^{r-1}\\
\end{pmatrix}
\]
Since $w_0/z_0 \ne 0$, each $r\times r$ minor of this $(r+1)\times r$ matrix is invertible: indeed, these minors are Vandermonde matrices, whose determinant is a product of factors of the form $(\zeta^i w_0/z_0 - \zeta^j w_0/z_0)$ with $i \neq j$. Hence, the code defined by $\gamma$ has locality~$r$.

The dimension of the code is
\[
\begin{split}
k 	&= \#S(m) + \#S(m-1) + \cdots + \#S(m-r+1)  - \dim(\ker \gamma) \\
	&= \binom{m+2}{2} + \binom{m+1}{2} + \cdots \binom{m-r+3}{2}- \dim(\ker \gamma).
\end{split}
\]
We shall discuss the minimum distance in specific cases below.

\subsection{Cubic surfaces: \texorpdfstring{$r=2$}{r=2}}

In the above framework, setting $r=2$ we consider smooth surfaces of the form
\[
X\colon\qquad w^3 = f_3(x,y,z),
\]
where $f_3(x,y,z)$ is a homogeneous cubic polynomial.  A smooth cubic surface is a del Pezzo surface of degree $3$. By~\cite[Theorem~1.1]{BFL}, this implies the inequalities
\[
q^2 - 2q + 1 \leq \#X(K) \leq q^2 + 7q + 1
\]
(recall our standing assumption that $3 \mid q - 1$, and hence $q \neq 2, 3$ or $5$).  Since $X$ is smooth, the Jacobian criterion shows that $C$ is smooth as well, and since $\phi|_{\phi^{-1}(C)}\colon \phi^{-1}(C) \to C$ is an isomorphism, we deduce that $\phi^{-1}(C)$ is a smooth curve of genus $1$. We conclude that
\[
q - 2\sqrt{q} + 1 \leq \#\phi^{-1}(C)(K) \leq q + 2\sqrt{q} + 1.
\]
The complement $Z$ of $U$ in $\PP^3$ is a plane; hence, its intersection with $X$ is a plane cubic that is not necessarily smooth, giving
\[
q - 2\sqrt{q} + 1 \leq \#(X \cap Z)(K) \leq 3q.
\]
The intersection $Z\cap \phi^{-1}(C)$ can contain at most three $K$-points, because its isomorphic image in $\PP^2$ is the intersection of $C$ with the line $z = 0$.  
Together, the above estimates give the following bounds for the length of the code $\gamma$:
\[
q^2 - 6q - 2\sqrt{q} \leq n\leq q^2 + 5q + 4\sqrt{q} + 2. 
\]
\begin{remark}
The above bounds for $n$ are probably not sharp.  We were unable, for example, to produce a cubic surface $X$ with $q^2 + 7q + 1$ $K$-points such that $\#\phi^{-1}(C)(K) = \#(X \cap Z)(K) = q - 2\sqrt{q} + 1$, which one would need to prove the upper bound for $n$ is sharp.
\end{remark}

The dimension of the code is
\[
\begin{split}
k 	&= \#S(m) + \#S(m-1)  - \dim(\ker \gamma) \\
	&= \binom{m+2}{2} + \binom{m+1}{2} - \dim(\ker \gamma) \\
	&= (m+1)^2 - \dim(\ker \gamma).
\end{split}
\]
To estimate the distance of the code, we compute the number of $K$-points of the intersection $\{z^{m+1}\cdot f_{\aa,\bb} = 0\} \cap X$.  The equation $z^{m+1}\cdot f_{\aa,\bb} = 0$ can be used to eliminate $w$ from the equation of $X$, leaving a single homogeneous equation in the variables $x$, $y$ and $z$ of degree $3m$, which we can think of as a plane curve (not necessarily smooth) of degree $3m$.  Serre's bound~\cite{SerreBound} tells us that this curve has at most $(3m)q + (q + 1) = (3m + 1)q + 1$ points, and hence
\begin{equation}
\label{eq:lowerbounddcubics}
d \geq n - (3m + 1)q - 1.
\end{equation}

We continue with a number of examples that illustrate the general construction above.
Some of the examples result in codes with the parameters $(n=(r+1)s, k=rs,2)$ and locality $r$, where $s$ is some positive number. Such codes can generally be obtained as $s$-fold repetitions of an $(r+1,r,2)$ single parity-check code, and their parameters meet the 
Singleton-type bound ~\eqref{eq:singleton} with equality. 

\begin{example}
\label{ex:smallcubic}
We begin with a simple example that produces a quaternary $(9,6,2)$-code with locality $2$ that meets the Singleton-type bound. Even though the resulting
code can be constructed without using surfaces, we include it because it provides a clear perspective of the general construction method.

We work over the finite field $\FF_4 = \{0,1,a,a^2\}$. Let
\[
f_3(x,y,z) := ax^3 + x^2y + axy^2 + ay^3 + a^2x^2z + a^2xyz + a^2xz^2 + ayz^2 + az^3,
\]
and consider the cubic surface
\[
	X\colon\qquad w^3 = f_3(x,y,z).
\]
The projection $\varphi\colon X\to \PP^2$, $[x,y,z,w] \mapsto [x,y,z]$ is branched along the plane curve
\[
	C\colon\qquad f_3(x,y,z) = 0.
\]
The curve $C$ has $4$ points over $\FF_4$. Of the remaining $17$ points in $\PP^2(\FF_4)$, only $3$ belong to $\varphi(X(\FF_4))$ but lie outside $C(\FF_4)$ and the line $z = 0$ in $\PP^2$.  They are
\[
Q_1 := [a^2 , a , 1], \quad Q_2 := [1 , a^2 , 1],\quad\text{and}\quad Q_3 := [0 , a^2 , 1].
\] 
The fibers of $\varphi$ over these points (i.e., the helper sets of the code) are
\begin{equation}
\label{eq:helpersets_ex}
\begin{split}
\varphi^{-1}(Q_1) &= \{[a^2 , a , 1, 1], [a^2 , a , 1, a], [a^2 , a , 1, a^2]\} \\
\varphi^{-1}(Q_2) &= \{[1 , a^2 , 1, 1], [1 , a^2 , 1, a], [1 , a^2 , 1, a^2]\} \\
\varphi^{-1}(Q_3) &= \{[0 , a^2 , 1, 1], [0 , a^2 , 1, a], [0 , a^2 , 1, a^2]\} 
\end{split}
\end{equation}
These fibers together give the points we use for the evaluation code.  Thus, the length of the code is $n = 9$.

We take $m = 2$ in the general construction above. For
\[
\aa = (a_1,a_2,a_3,a_4,a_5,a_6) \in (\FF_4)^6 \quad\text{and}\quad \bb = (b_1,b_2,b_3) \in (\FF_4)^3,
\]
we let $f_{\aa,\bb} \in K(\PP^2)$ be the function
\[
f_{\aa,\bb} = \frac{1}{z^2}\cdot(a_1x^2 + a_2xy + a_3xz + a_4y^4 + a_5yz + a_6z^2) + \frac{w}{z^{3}}\cdot(b_1x + b_2y + b_3z).
\]
and we define a map $\gamma\colon (\FF_4)^{6} \times (\FF_4)^{3} \to (\FF_4)^9$ by $\gamma(\aa,\bb) = (f_{\aa,\bb}(P)	)$, where $P \in \PP^3(\FF_4)$ ranges over the points in~\eqref{eq:helpersets_ex}.  The map $\gamma$ has a $3$-dimensional kernel, so the dimension of the code we get from it is $k = 9 - 3 = 6$.  We computed the $4^{9 - 3} - 1= 4095$ nonzero elements of the image of $\gamma$, and found the minimum distance of the code to be $d = 2$.  Note that the lower bound~\eqref{eq:lowerbounddcubics} gives $d \geq -20$, which is quite poor in comparison to the actual distance of the code.
\end{example}

\begin{example}
\label{ex:q^2+2}
The surface
\[
X\colon \qquad w^3 = xy^2 + y^3 + a^2x^2z + xyz + ay^2z + a^2z^3
\]
over $\FF_4 = \{0,1,a,a^2\}$ can be used to produce a $(18,11,3)$-code with locality $2$ code that meets the Singleton-type~\eqref{eq:singleton} bound.

In this case, the curve $C$ also has $4$ points over $\FF_4$. Of the remaining $17$ points in $\PP^2(\FF_4)$, $6$ belong to $\varphi(X(\FF_4))$ but lie outside $C(\FF_4)$ and the line $z = 0$ in $\PP^2$.  They are
\[
\begin{split}
Q_1 &:= [a^2 , 1 , 1], \quad Q_2 := [1 , a , 1],\quad Q_3 := [a^2 , a , 1]\\
Q_4 &:= [a , a^2 , 1], \quad Q_5 := [0 , a^2 , 1],\quad Q_6 := [a , 0 , 1].
\end{split}
\] 
The fibers of $\varphi$ over these points are
\begin{equation*}
\begin{split}
\varphi^{-1}(Q_1) &= \{[a^2 , 1 , 1, 1], [a^2 , 1 , 1, a], [a^2 , 1 , 1, a^2]\} \\
\varphi^{-1}(Q_2) &= \{[1 , a , 1, 1], [1 , a , 1, a], [1 , a , 1, a^2]\} \\
\varphi^{-1}(Q_3) &= \{[a^2 , a , 1, 1], [a^2 , a , 1, a], [a^2 , a , 1, a^2]\} \\
\varphi^{-1}(Q_4) &= \{[a , a^2 , 1, 1], [a , a^2 , 1, a], [a , a^2 , 1, a^2]\} \\
\varphi^{-1}(Q_5) &= \{[0 , a^2 , 1, 1], [0 , a^2 , 1, a], [0 , a^2 , 1, a^2]\} \\
\varphi^{-1}(Q_6) &= \{[a , 0 , 1, 1], [a , 0 , 1, a], [a , 0 , 1, a^2]\} 
\end{split}
\end{equation*}
These fibers together give the points we use for the evaluation code, giving a code of length $n = 18$.

Taking $m = 3$ in the general construction above gives a map
\[
\gamma\colon (\FF_4)^{10} \times (\FF_4)^{6} \to (\FF_4)^{18}
\]
that has a $5$-dimensional kernel. Hence, the dimension of the resulting code is $k = 16 - 5 = 11$.  We computed the $4^{11} - 1= 4,194,303$ nonzero elements of the image of $\gamma$, and found the minimum distance of the code to be $d = 3$.
\end{example}


\begin{example}
Up to $\FF_4$-isomorphism, there are $13$ smooth plane cubic curves $f_3(x,y,z) = 0$ over $\FF_4$.  Each curve gives rise to a cubic surface of the form $w^3 = f_3(x,y,z)$. In Table~\ref{ta:genus1}, we have chosen models of the $13$ plane cubics that maximize the length of the associated code on the cubic surface.
\begin{table}[t]
\begin{tabular}{ clll }
\toprule
Surface & $m$ & $(n,k,d)$ & SG \\ 
\midrule
{$w^3 = ax^3 + x^2y + a^2xy^2 + a^2x^2z + a^2y^2z + xz^2 + yz^2 + z^3$} & $3$ & $(30,15,3)$ & 6 \\
 & $4$ & $(30,19,2)$ & 1\\[0.5ex]
{$w^3 = a^2x^3 + x^2y + axy^2 + ax^2z + ay^2z + xz^2 + yz^2 + z^3$} & $3$ & $(30,15,3)$ & 6 \\
 & $4$ & $(30,19,2)$ & 1\\[0.5ex]
{$w^3 = x^2y + xy^2 + x^2z + y^2z + xz^2 + yz^2 + z^3$} & $3$ & $(30,15,3)$ & 6 \\
 & $4$ & $(30,19,2)$ & 1\\[0.5ex]
{$w^3 = a^2x^3 + ax^2y + xy^2 + x^2z + axyz + y^2z + a^2xz^2$} & $3$ & $(27,15,3)$ & 3\\
 & $4$ & $(27,18,2)$ & 0\\[0.5ex]
{$w^3 = ax^3 + x^2y + xy^2 + x^2z + xyz + a^2xz^2 + z^3$} &  $3$ & $(27,15,3)$ & 3\\
 & $4$ & $(27,18,2)$ & 0\\[0.5ex]
{$w^3 = a^2x^2y + xy^2 + a^2xyz + y^2z + z^3$} &  $3$ & $(27,15,3)$ & 3\\
 & $4$ & $(27,18,2)$ & 0\\[0.5ex]
{$w^3 = ax^3 + x^2y + xy^2 + a^2x^2z + xz^2 + z^3$} & $3$ & $(24,14,3)$ & 2\\
 & $4$ & $(24,16,2)$ & 0\\[0.5ex]
{$w^3 = a^2x^2y + xy^2 + a^2x^2z + a^2xyz + y^2z + xz^2 + a^2z^3$} & $3$ & $(21,13,2)$ & 1\\
 & $4$ & $(21,14,2)$ & 0\\[0.5ex]
{$w^3 = a^2x^3 + x^2y + xy^2 + x^2z + xyz + xz^2 + z^3$} & $3$ & $(21,13,2)$ & 1\\
 & $4$ & $(21,14,2)$ & 0\\[0.5ex]
{$w^3 = ax^3 + x^2y + xy^2 + x^2z + xyz + xz^2 + z^3$} & $3$ & $(21,13,2)$ & 1 \\
 & $4$ & $(21,14,2)$ & 0\\[0.5ex]
{$w^3 = a^2x^3 + ax^2y + xy^2 + x^2z + xz^2 + z^3$} & $3$ & $(18,11,2)$ & 1\\
 & $4$ & $(18,12,2)$ & 0\\[0.5ex]
{$w^3 = ax^3 + a^2x^2y + xy^2 + x^2z + xz^2 + z^3$} & $3$ & $(18,11,2)$ & 1\\
 & $4$ & $(18,12,2)$ & 0\\[0.5ex]
{$w^3 = ax^3 + a^2x^2y + a^2xy^2 + ay^3 + x^2z + y^2z + xz^2 + yz^2 + z^3$} & $3$ & $(12,7,3)$ & 0\\
 & $4$ & $(12,8,2)$ & 0\\
\bottomrule
\end{tabular}
\vskip1ex
\caption{Maximal length codes arising from all $\FF_4$-isomorphism classes of cubic surfaces over $\FF_4$ of the form $w^3 = f_3(x,y,z)$. All codes have locality $r = 2$. SG stands for Singleton gap, i.e., the difference between the value of the distance and the right-hand side of \eqref{eq:singleton}.}
\label{ta:genus1}
\end{table}
The surfaces of Examples~\ref{ex:smallcubic} and~\ref{ex:q^2+2} are both isomorphic to the surface in the fourth row from the bottom of Table~\ref{ta:genus1}, but the models we chose for the curve $C$ in these examples give rise to codes of  shorter length.
\end{example}

\begin{example}
\label{ex:F7}
The surface
\[
X\colon \qquad w^3 = 6x^3 + 5xy^2 + y^3 + 2x^2z + 3xyz + 4y^2z + 4xz^2 + 6yz^2
\]
over $\FF_7$ can be used to produce a $(48,31,3)$-code with locality $2$ code that meets the Singleton-type~\eqref{eq:singleton} bound (note that $n = q^2 - 1$).  There are $16$ points in $\PP^2(\FF_7)$ that belong to $\varphi(X(\FF_7))$ but lie outside the branch curve and the line $z = 0$ in $\PP^2$, which is why the resulting code has length $n = 48$.  We take $m = 5$ in our construction, and obtain a map
\[
\gamma\colon (\FF_7)^{21} \times (\FF_7)^{15} \to (\FF_7)^{48}
\]
that has a $5$-dimensional kernel. Hence, the dimension of the resulting code is $k = 36 - 5 = 31$.  Using \textsc{Magma}~\cite{magma}, we found the minimum distance of the code to be $d = 3$.
\end{example}

\subsection{K3 surfaces: \texorpdfstring{$r=3$}{r=3}}

In the above framework, when setting $r=3$, we consider smooth surfaces of the form
\[
X\colon\qquad w^4 = f_4(x,y,z),
\]
where $f_4(x,y,z)$ is a homogeneous polynomial of degree $4$.  A smooth quartic surface in $\PP^3$ is a K3 surface.  These surfaces can also be used to construct codes that meet the Singleton-type bound~\eqref{eq:singleton}, yet their minimum distance is $d > 2$ and they have large length compared to the size of the alphabet $(n \sim q^2)$. These codes have locality $3$.

\begin{example}
\label{ex:K3}
We construct a $(24,17,3)$-code with locality $3$ over the field $\FF_5$ (note that $n = q^2 -1$ in this case.).  The information rate for this code is $17/24$, and it meets the Singleton-type bound~\eqref{eq:singleton}.

Let
\[
\begin{split}
f_4(x,y,z) &:= 3x^4 + x^3y + 4x^2y^2 + 4xy^3 + 4y^4 + x^3z + 2x^2yz  \\
&\quad + xy^2z+ 4y^3z + 3x^2z^2 + xyz^2 + y^2z^2 + 2xz^3 + 3z^4,
\end{split}
\]
and consider the K3 surface
\[
	X\colon\qquad w^4 = f_4(x,y,z).
\]
The projection $\varphi\colon X\to \PP^2$, $[x,y,z,w] \mapsto [x,y,z]$ is branched along the plane curve
\[
	C\colon\qquad f_4(x,y,z) = 0.
\]
The curve $C$ has $2$ points over $\FF_5$. Of the remaining $29$ points in $\PP^2(\FF_4)$, only $6$ belong to $\varphi(X(\FF_5))$ but lie outside $C(\FF_5)$ and the line $z = 0$ in $\PP^2$.  They are
\[
\begin{split}
Q_1 &:= [3 , 0 , 1], \quad Q_2 := [4 , 0 , 1],\quad Q_3 := [4 , 2 , 1]\\
Q_4 &:= [3 , 3 , 1], \quad Q_5 := [4 , 3 , 1],\quad Q_6 := [3 , 4 , 1].
\end{split}
\] 
The fibers of $\varphi$ over these points (i.e., the helper sets of the code) are
\begin{equation*}
\begin{split}
\varphi^{-1}(Q_1) &= \{[3 , 0 , 1, 1], [3 , 0 , 1, 2], [3 , 0 , 1, 3], [3 , 0 , 1, 4]\} \\
\varphi^{-1}(Q_2) &= \{[4 , 0 , 1, 1], [4 , 0 , 1, 2], [4 , 0 , 1, 3], [4 , 0 , 1, 4]\} \\
\varphi^{-1}(Q_3) &= \{[4 , 2 , 1, 1], [4 , 2 , 1, 2], [4 , 2 , 1, 3], [4 , 2 , 1, 4]\} \\
\varphi^{-1}(Q_4) &= \{[3 , 3 , 1, 1], [3 , 3 , 1, 2], [3 , 3 , 1, 3], [3 , 3 , 1, 4]\} \\
\varphi^{-1}(Q_5) &= \{[4 , 3 , 1, 1], [4 , 3 , 1, 2], [4 , 3 , 1, 3], [4 , 3 , 1, 4]\} \\
\varphi^{-1}(Q_6) &= \{[3 , 4 , 1, 1], [3 , 4 , 1, 2], [3 , 4 , 1, 3], [3 , 4 , 1, 4]\}
\end{split}
\end{equation*}
These fibers together give the points we use for the evaluation code, giving a code of length $n = 24$.

Taking $m = 4$ in the general construction above gives a map
\[
\gamma\colon (\FF_5)^{15} \times (\FF_5)^{10} \times (\FF_5)^{6} \to (\FF_5)^{24}
\]
that has a $14$-dimensional kernel. Hence, the dimension of the resulting code is $k = 31 - 14 = 17$.  Using {\tt magma}~\cite{magma}, we found the minimum distance of the code to be $d = 3$.

\end{example}

\subsection{Surfaces of general type: \texorpdfstring{$r=4$}{r=4}}

Consider smooth surfaces of the form
\[
X\colon\qquad w^5 = f_5(x,y,z),
\]
where $f_5(x,y,z)$ is a homogeneous polynomial of degree $5$.  A smooth quintic surface in $\PP^3$ is a surface of general type.  In our construction, we require that $r + 1 \mid q - 1$.  Since $r = 4$, the smallest possible $q$ we can use is $11$.

\begin{example}
\label{ex:F11}
We produce a $(110,87,3)$-code with locality $4$ over the field $\FF_{11}$.  The code meets the Singleton-type bound~\eqref{eq:singleton}. Let
\[
\begin{split}
f_5(x,y,z) &:= 9x^5 + 2x^4y + x^3y^2 + 5x^2y^3 + 6xy^4 + 4y^5 \\
&\quad + 6x^4z + 3x^3yz + 3x^2y^2z + 8xy^3z + 2y^4z  \\
&\quad + 10x^3z^2 + 3x^2yz^2 + 7xy^2z^2 + 6y^3z^2  \\
&\quad + 3x^2z^3 + 5xyz^3 + 8y^2z^3 + 6xz^4 + 6yz^4
\end{split}
\]
and consider the surface
\[
	X\colon\qquad w^5 = f_5(x,y,z).
\]
\end{example}
The plane curve $C$ given by $f_5(x,y,z) = 0$ has $13$ points, and of the remaining $120$ points in $\PP^2(\FF_{11})$, only $22$ belong to $\varphi(X(\FF_{11}))$ but lie outside $C(\FF_{11})$ and the line $z = 0$ in $\PP^2$.  The fibers of morphism $\varphi$ above these points together give the points we use for the evaluation code, giving a code of length $n = 110$.  We use $m = 8$ in our general construction to construct the evaluation map
\[
\gamma\colon (\FF_{11})^{45} \times (\FF_{11})^{36} \times (\FF_{11})^{28} \times (\FF_{11})^{21} \to (\FF_{11})^{110},
\]
which has a $43$-dimensional kernel.  Hence, the dimension of the resulting code is $k = 130 - 43 = 87$.  Using {\tt magma}~\cite{magma}, we found the minimum distance of the code to be $d = 3$.

\begin{question}
\label{q:arbitrary_r}
Fix a positive integer $r$.  Does there always exist a smooth surface of the form $w^{r + 1} = f_{r+1}(x,y,z)$ over a finite field $\FF_q$, whose associated code as above meets the Singleton-type bound, has minimum distance $3$, length $n \sim q^2$, and locality $r$?
\end{question}


\section{Conclusion}\label{sec:conclusion}
In Table \ref{table2} we collect the parameters of the code families explicitly mentioned in this paper. We did not make an attempt to make an exhaustive list of short LRCs from curves, or optimize their parameters, which would hardly be possible given the variety of the constructions studied above. Codes constructed from 
cubic surfaces are listed in Table~\ref{ta:genus1} above and not included here, although the surfaces from Examples~\ref{ex:smallcubic} and~\ref{ex:q^2+2} are (these examples meet the Singleton bound \eqref{eq:singleton} with equality). 

{\small 
\begin{table}[t]
\begin{tabular}{ccccclccc}
\toprule
\multirow{2}{*}{$q$} & \multirow{2}{*}{$n$} &\multirow{2}{*}{$k$}    & Designed  & \multirow{2}{*}{$r$} & \multirow{2}{*}{Remarks}   &\multirow{2}{*}{SG} & \multirow{2}{*}{Curve or Surface} & \multirow{2}{*}{Reference} \\
& & & distance & & & & & \\
\midrule
 4   
     &18  &11    &3           &2   &                   &0   &cubic surface                  &Ex.~\ref{ex:q^2+2}\\[0.5ex]
 5   &24  &17    &3           &3   &                   &0   &quartic K3 surface             &Ex.~\ref{ex:K3}\\[0.5ex]
 7   &20  &$3t$  &$20-4t$     &3   &$1\le t\le 4$      &2   &plane quartic                  &Ex.~\ref{ex:quartic7}\\
     &48  &$31$  &$3$         &2   &                   &0   &cubic surface                  &Ex.~\ref{ex:F7}\\[0.5ex]
 8   &24  &$3t$  &$24-4t$     &3   &$1\le t\le 5$      &2   &plane quartic                  &Ex.~\ref{ex:quartic8}\\[0.5ex] 
 11  &110 &87    &3           &4   &                   &0   &quintic surface                &Ex.~\ref{ex:F11}\\[0.5ex]
 16  &36  &$3t$  &$36-4t$     &3   &$1\le t\le 8$      &2   &plane quartic                  &Ex.~\ref{ex:quartic16}\\
     &45  &$2t$  &$42-2t$     &2   &$1\le t\le 13$     &2   &genus-$7$ curve                &Ex.~\ref{ex:triplecoverF16}\\ 
     &63  &$2t$  &$51-3t$     &2   &$1\le t\le 16$     &2   &genus-$6$ Hermitian curve      &Ex.~\ref{ex:hermitian16}\\[0.5ex]
 31  &56  &$3t$  &$56-4t$     &3   &$1\le t\le 13$     &2   &genus-$3$ hyperelliptic curve  &Ex.~\ref{ex:hyperelliptic31}\\
     &60  &$3t$  &$60-4t$     &3   &$1\le t\le 14$     &2   &plane quartic                  &Ex.~\ref{ex:quartic31}\\[0.5ex]
 32  &40  &$3t$  &$37-4t$     &3   &$1\le t\le 9$      &5   &elliptic curve                 &Ex.~\ref{ex:firstECGF32}\\
     &42  &$2t$  &$40-3t$     &2   &$1\le t\le 13$     &4   &elliptic curve                 &Ex.~\ref{ex:thirdECGF32}\\
     &64  &$3t$  &$64-4t$     &3   &$1\le t\le 15$     &2   &plane quartic                  &Ex.~\ref{ex:quartic32}\\     
     &87  &$2t$  &$84-3t$     &2   &$1\le t\le 27$     &2   &genus-$7$ curve                &Ex.~\ref{ex:triplecoverF32}\\[0.5ex]
 64  &171 &$2t$  &$168-3t$    &2   &$1\le t\le 55$     &5   &genus-$7$ curve                &Ex.~\ref{ex:triplecoverF64}\\      
\bottomrule
\end{tabular}
\vskip1ex
\caption{\small Some examples of codes constructed in this paper.  SG stands for Singleton gap, i.e., the
difference between the value of the distance and the right-hand side of \eqref{eq:singleton}. Additional codes over $\FF_4$ with locality 2 are
listed in Table~\ref{ta:genus1} above.}\label{table2}
\end{table}
}

Note that several examples in the table have parameters meeting the Singleton bound. In the classical case of $r=k$ the bound
\eqref{eq:singleton} reduces to the inequality $d\le n-k+1,$ and codes that meet this bound are called MDS codes. 
It is easy to construct MDS codes of any length $n$ with distance $d=1$, $2$, $n$, and such codes are called trivial. At the same
time, all the known nontrivial MDS codes have length $n\le q+1,$ except for even $q$ and $k=3$ or $q-1$ when there are codes
with $n=q+2$. The famous MDS conjecture asserts that longer MDS codes over any field $\Fq$ do not exist \cite[pp.~326ff.]{MaSl1977},
\cite[Sec.~11.4]{Roth06}. Recently this statement
was proved in some cases \cite{Ball12,BallBeule12}, but the full MDS conjecture remains open. In view of
the above examples, generalizing this conjecture to codes with locality constraints appears to be a difficult question.
In particular, there exists a locally recoverable ``MDS" code of length $18$ over $\FF_4$ whose length $n=q^2 + 2$, a code of length $24$ over $\FF_5$ whose length $n=q^2 - 1$, and a code of length $110$ over $\FF_{11}$ whose length is $n = q^2 - q$. Moreover, Question~\ref{q:arbitrary_r} is meant to suggest that there may exist a family of optimal LRCs with $n$ in the order 
of $q^2$ for all values of locality $r$. Thus, extending the MDS conjecture to codes with locality requires much more detailed
understanding of geometric or combinatorial conditions for equality in the bound~\eqref{eq:singleton}.

Finally, we compare codes from Table \ref{table2} with the parameters of best known linear codes with no 
locality assumptions taken from \cite{Grassl16}, which lists these parameters for $q=2,\dots,9$. 

\begin{center}
\begin{tabular}{ccccc}
\toprule
$q$      &$n$      &$k$      &$d$(LRC)     &Best known $d$ \\
\midrule
4         &18     &11     &3          &3\\
5         &24     &17     &3          &6\\       
7         &20     &9      &8          &9\\
8         &24     &12     &8          &10\\
\bottomrule
\end{tabular}
\end{center}

We note again that we made no attempt to optimize the distance of codes in Table \ref{table2}.

\subsection*{Acknowledgments}
We are grateful to the organizers of the workshop on
``Algebraic Geometry for Coding Theory and Cryptography,'' 
February 22--26, 2016, at the Institute for Pure and Applied Mathematics
on the campus of the University of California, Los Angeles, for bringing us together and providing the environment to pursue the research
presented in this paper. A.~B.~was supported by NSF grants CCF142955 and CCF1618603. A.~V.-A.~was supported by NSF CAREER grant DMS-1352291.


\bibliographystyle{hplaindoi} 
\bibliography{LRC}

\end{document}